\documentclass[conference]{IEEEtran}
\usepackage{graphicx, subfigure}
\usepackage{multicol}
\usepackage{caption}
\usepackage{mathptmx}
\usepackage{amsmath,amssymb,amsthm}

\usepackage{amssymb}
\usepackage{subeqnarray}
\usepackage{cite}
\usepackage{float}
\usepackage{algorithm}
\usepackage{algorithmic}
\ifCLASSINFOpdf
\else
\fi

\begin{document}
%
\title{Characterization of Driver Nodes of Anti-Stable Networks}

\author{\IEEEauthorblockN{Ram Niwash Mahia}
\IEEEauthorblockA{Indian Institute of Technology Jodhpur\\
Rajasthan, INDIA-342011\\
Email: ramyashu@iitj.ac.in}
\and
\IEEEauthorblockN{Deepak Fulwani}
\IEEEauthorblockA{Indian Institute of Technology Jodhpur\\
Rajasthan, INDIA-342011\\
Email: df@iitj.ac.in}
\and
\IEEEauthorblockN{Mahaveer Singh}
\IEEEauthorblockA{Indian Institute of Technology Jodhpur\\
Rajasthan, INDIA-342011\\
Email: pg201373010@iitj.ac.in}}


%


\maketitle

\begin{abstract}
A  controllable  network can  be  driven  from  any  initial  state  to  any  desired  state using driver nodes. A set of driver nodes to control a network is not unique. It is important to characterize these driver nodes and select the right driver nodes. The work discusses theory and algorithms to select driver node such that largest region of attraction can be obtained considering limited capacity of driver node and with unstable eigenvalues of adjacency matrix. A network which can be controllable using one driver node is considered. Nonuniqueness of driver node poses a challenge to select right driver node when multiple possibilities exist. The work addresses this issue.
\end{abstract}

\begin{keywords}
Complex network, Characterization of Driver Nodes, Region of Attraction.
\end{keywords}



%
\IEEEpeerreviewmaketitle

%


\section{Introduction}
\par Twenty first century has been witnessing unprecedented growth in use of networked systems. Application domains of networked systems are as diverse as social systems to bilogical networks. A network can be defined as an entity comprising of nodes and edges where each node represents an entity such as genes in a biological network, sensors in a detection system, vehicles in a traffic system or persons/individuals in a social network; while the edges denote connection or interactions between the nodes \cite{liu2011controllability}. The state of a node can be described as position coordinates of a sensor or a robot, an opinion of a person, expression of protein by a gene etc \cite{yan2012controlling}. A complex network system is composed of several nodes, which can interact with each other as well as share information. Each node has a state variable to represent its state values. These nodes, together, form a system, which performs tasks collectively.

\par Recent control theoric analysis of complex networks provides a great deal of insight about the behaviour of complex network \cite{liu2011controllability,yan2012controlling,controllabilityPRL2007, controllabilitytransitionPRL2013, networkcontrollabilityPRL2014}. Controllability property is investigated by modeling complex network as a linear time invariant system; adjacency matrix is used as system matrix. A recent contribution has been made by Liu et al.\cite{liu2011controllability}, who developed a minimum input theory to efficiently characterize the structural controllability of directed networks using a minimum set of driver nodes to control of the system. The complex network system can also be controlled by the driver node with minimum control energy of the input which is an unavoidable and important issue \cite{Bullo2013controllability}. In \cite{Bullo2013controllability}, authors proposed the trade-off between the driver nodes and  the control energy as a function of the network dynamics using the smallest eigenvalue of the Controllability Gramian. One of the challenges in the control of complex network systems is to find a set of  right nodes (physical systems), so that controlling these nodes eventually leads to control of the entire network in a desired manner. In recent years there have been quite a few work devoted to this problem \cite{liu2011controllability,yan2012controlling,rahmani2009controllability,controllabilityPRL2007, controllabilitytransitionPRL2013, networkcontrollabilityPRL2014, optimizingcontrollabilityPRL2012}.

\par While considering the problem of controlling a complex network, it is important to consider limitations of driver nodes. Driver nodes can not have infinite (unlimited) actuating capability i.e., maximum input a driver node can provide is limited. When actuator has limited capacity, control objectives may not be achieved if this limitation is not considered  \textit{a priori}. This problem becomes more complex when adjacency matrix has unstable eigenvalues; even stability can not be guaranteed in this situation. With unstable eigenvalues of adjacency matrix of complex network, region in state space, where stability is guaranteed, is finite when actuator (driver node) limitation is considered. We denote this region as region of attraction (ROA). This ROA depends on choice of driver node and control law. This work characterizes driver nodes and proposes a theory and algorithms such that the right driver node, among many, can be found to maximize region of attraction. In $n$ dimensional state space, region is described by ellipsoid. In context of linear system with actuator saturation, many excellent work exist see \cite{gutman1982new,Thu2001actuatorsaturationbook,lauvdal1997stabilization} and the references therein. Number of driver nodes required to control the network is fixed i.e., set of driver nodes has fixed number of nodes, however, driver nodes are not unique. This raises a very important question about the choice of driver nodes when there are many nodes qualify for driver nodes. In this work we address two important issues the first issue is regarding selection of driver node and in second we consider driver node limitation. We propose a theory and algorithm to select driver node such that region of attraction corresponding to this driver node maximizes. In this work to obtain largest possible region with a given control law, the low and high gain (LHG) technique \cite{gutman1982new,turner2001guaranteed, lowhighreview2009,saberi1996control} is used. Our emphasis on selection of right driver nodes rather than on control law. To the best of authors knowledge this work does not exist.

\par Rest of the paper is organized as follows. Starting with an introduction in section I, the section II describes the modeling of the networks with saturated actuator. Section III presents the theoretical results to characterize the maximum volume of the invariant ellipsoids of anti-stable system and transformed anti-stable subsystem. The simulation results of the two networks and their discussion are included in section IV. Finally, the concluding remark is included in the  section V. 


\newtheorem{thm}{\textbf{Theorem}}
\newtheorem{rem}{\textbf{Remark}}
\newtheorem{lm}{\textbf{Lemma}}
\newtheorem{assum}{\textbf{Assumption}}
\newtheorem{condi}{\textbf{Condition}}
\newtheorem{defi}{\textbf{Definition}}
\newtheorem{propos}{\textbf{Proposition}}
\newtheorem{corol}{\textbf{Corollary}}
\newtheorem{example}{\textbf{Example}}
\section{Modeling of Network with Saturated Actuator}
A network with $n$ nodes can be represented by a graph $\mathrm{G}: =(\mathrm{V},\mathrm{E})$, where $\mathrm{V} :=\{1,2, \ldots, n\}$ and $\mathrm{E} \subseteq \mathrm{V} \times \mathrm{V}$ are the set of nodes and the set of edges respectively. The weighted adjacency matrix of the network $\mathrm{G}$ is described as $\mathrm{A}(\mathrm{G})$, where the elements are $a_{ij} \in \mathbb{R} $ if node $i$ and $j$ adjacent and $0$ otherwise. Let us assume that network $\mathrm{G}$ is independently controllable by each of the nodes from the set $\bar{U} = \{u_{1},u_{2}, \ldots, u_{p}\}$ i.e. each input $u_{i}, ~ where ~ i=1,2,\ldots,p$, acting alone can control the entire network.
\par Consider a complex network with input $u_{i}$, let us designate corresponding input matrix as $B_{i}$, with this, dynamics of network with bounded control can be written as
\begin{eqnarray}
\centering
\frac{d\textbf{x}(t)}{dt}=A(G)\textbf{x}(t)+B_{i}sat(u_{i}(t))            \label{eq.1}
\end{eqnarray}
where $A(G) \in \mathbb{R}^{n\times n}$ adjacency matrix of the network $G$, $B_{i} \in \mathbb{R}^{n\times 1}$ input matrix corresponding to  input $u_{i}$, $\textbf{x} \in \mathbb{R}^{n\times 1}$ state vector (n-vector), $u_{i} \in \mathbb{R} $ control signal (scalar); It should be noted that any $u_{i}$ can control the network. The standard saturation function $`sat'$ is defined as $ sat \colon \mathbb{R} \rightarrow \mathbb{R} $, i.e., $sat(u_{i}(t))= sign(u_{i}(t)) min \{u_{i,max},~ \lvert u_{i}(t) \rvert\}$ where $u_{i,max}>0$, $\forall ~ i=1,2,\ldots,p$ is the saturation limit control input $u_{i}$. Without loss of generality let us assume $u_{i,max}=1$. The state vector is defined as $\textbf{x}=[x_{1} ~ x_{2} \ldots x_{n}]^{\mathrm{T}}$. \\
Let us define control law for network (\ref{eq.1}) as
\begin{equation}
\centering
u_{i}:=-K_{i}\textbf{x}=-R^{-1}B_{i}^{T}P_{i}\textbf{x}   \label{eq.2}
\end{equation}
Where $R>0$ is a chosen matrix and the matrix $P_{i}>0$ is obtained by solving the following Ricatti equation for some $Q>0$
\begin{equation}
\centering
A^{T}P_{i}+P_{i}A-P_{i}B_{i}R^{-1}B_{i}^{T}P_{i}+Q=0   \label{eq.3}
\end{equation}
\par Now, the control law (\ref{eq.2}) is divided into two equal parts, low gain and high gain i.e. $u_{i,L}=-\frac{1}{2}R^{-1}B_{i}^{T}P_{i}\textbf{x}$ and $u_{i,H}=-\frac{1}{2}R^{-1}B_{i}^{T}P_{i}\textbf{x}$ respectively, where the low gain part is not allowed to saturate but the overall control law  is allowed \cite{turner2001guaranteed}.
\begin{equation}
\centering
u_{i}:=u_{i,L}+u_{i,H}   \nonumber
\end{equation}
\par Under the control law (\ref{eq.2}), the network is asymptotically stable for all states lie in the invariant ellipsoid $\varepsilon(P_{i},\delta_{i})$, \cite{gutman1982new} given as
 \begin{equation}
 	\begin{aligned}
 	\varepsilon(P_{i},\delta_{i})=\{\textbf{x}:\textbf{x}^{T}P_{i}\textbf{x}\leq\delta_{i}\} 														\label{eq.4}
 	\end{aligned}
 	\end{equation}
 	and $\delta_{i}>0$ is defined as
 \begin{equation}
  	\begin{aligned}								\label{eq.5}
  	\delta_{i} := \frac{4r_{i}^{2}}{B_{i}^{T}P_{i}B_{i}}
  	\end{aligned}
  	\end{equation}
  Where $r_{i}$ is the $i^{th}$ diagonal entry of $R>0$ and it is assumed $r_{i}=1$. Equation (\ref{eq.5}) gives the radius of the invariant ellipsoid $\varepsilon(P_{i},\delta_{i})$ and it is obtained by using the low gain part of the linear quadratic (LQ) control (\ref{eq.2}). The proof of (\ref{eq.5}) is given in \cite{turner2001guaranteed,gutman1982new,Henrion1999}. It is evident from (\ref{eq.5}), that with different input node (driver node) acting independently, the matrices $P_{i}$ and $B_{i}$ will change and this subsequently changes $\delta_{i}$ and corresponding region of attraction.

\section{Main Results}
This section is divided into three parts, first, we explain the problem statement, second, we derive the conditions for ROA of the anti-stable network with different control inputs and in third, we derive the conditions for ROA of the anti-stable subsystem network with different control inputs of the network (\ref{eq.1}).
\subsection{Problem Statement}
A complex network system of homogeneous nodes with a set of nodes, having capability to control the network acting alone, is considered. A set of driver nodes\textemdash each driver node can control the network independently\textemdash is considered. Furthermore, with different input, region where stability, with saturated input, is guaranteed may expand or shrink. This work presents theory and algorithms to find an input such that largest region of attraction can be achieved. The objective is to characterize an input such that ROA becomes largest in comparison to region obtained by each of the other driver nodes from $\bar{U}$ acting alone.

\subsection{Region of Attraction and Selection of Driver Node of an Anti-stable Network With Saturated Driver Noder}
This section discusses region of attraction and volume of invariant ellipsoid for the system defined in (\ref{eq.1}) with all eigenvalues of $A(G)$ are unstable.
Consider an anti-stable network (all eigenvalues of $A(G)$ are unstable) (\ref{eq.1}) with driver node saturation. The contractively invariant ellipsoids corresponding to input matrix $B_{i}$ of the network (\ref{eq.1}) depends on the positive definite matrix $P_{i}$.

\begin{lm} For the system defined in (\ref{eq.1}) with control input $u_{i}$, control law (\ref{eq.2}) and radius of the invariant ellipsoid defined in (\ref{eq.5}): \\
 (i) The volume of the invariant ellipsoid is given by
 	\begin{equation}
  		\centering
  		Vol_{i}=\frac{S(0,\sqrt{\delta_{i}})}{\sqrt{det(P_{i})}}
  												\label{eq.6}
  		\end{equation}
  		Where $S(0,\sqrt{\delta_{i}})$ is the volume of $n-$dimensional sphere $(n\geq~3)$ of radius $\sqrt{\delta_{i}}$ and center at origin, $det(P_{i})$ is the determinant of the positive definite matrix $P_{i}$. \\
 (ii) Largest region of attraction corresponds to $i^{th}$ driver node for which volume becomes largest i.e. \\
  	$Largest~~ volume=Max_{i}\{Vol_{i}\} ~~~ i=1,2,\ldots,p$
  	\end{lm}
  \begin{proof} To prove the first part, the outer boundaries of the contractively invariant ellipsoids $\varepsilon(P_{i},\delta_{i})=\{\textbf{x}:\textbf{x}^{T}P_{i}\textbf{x}\leq\delta_{i}\}$ of the linear anti-stable system (\ref{eq.1}) are obtained from the following equation
  \begin{equation}
  \centering
  \textbf{x}^{T}P_{i}\textbf{x}=~\delta_{i}
  											\label{eq.7}
  \end{equation}
  and the volume of the contractively invariant ellipsoid is given by
  \begin{equation}
  \centering
  Vol_{i}=\underbrace{\idotsint}_{x_{j_{min},j=1,..,n}}^{x_{j_{max}}}\underbrace{dx_{1}dx_{2}....dx_{n}}
  											\label{eq.8}
  \end{equation}
  Where $x_{j_{min}}$ and $x_{j_{max}}$ are the minimum and maximum limits of state vector $\textbf{x}$ in $j^{th}$ direction. \\
  Now assume $\textbf{x}=T_{i}\textbf{y}$, where $T_{i}>0$ is a positive definite matrix and $\textbf{y} \in \mathbb{R}^{n}$ state vector, then equation (\ref{eq.7}) can be rewritten as
  \begin{equation}
  \centering
  \textbf{y}^{T}T_{i}^{T}P_{i}T_{i}\textbf{y}=~\delta_{i}
  												\label{eq.9}
  \end{equation}
  The positive definite matrix $T_{i}$ is chosen in such a way that  $T_{i}^{T}P_{i}T_{i}=I_{n\times n}$. This implies that $T_{i}=((M_{i}\sqrt{D_{i}})^{-1})^{T}$, where $M_{i}$ is the eigenvector matrix of positive definite matrix $P_{i}$ and matrix $D_{i}$ is the diagonal matrix (elements are the eigenvalues of the matrix $P_{i}$).
  \par Using change of variables $\textbf{x}=T_{i}\textbf{y}$, we can obtain $\underbrace{dx_{1}~dx_{2}~\ldots~dx_{n}}=(\frac{1}{\sqrt{det(P_{i})}})\underbrace{dy_{1}~dy_{2}~\ldots~dy_{n}}$.
   Equation (\ref{eq.8}) become
  \begin{equation}
  \centering
  Vol_{i}=(\frac{1}{\sqrt{det(P_{i})}})\underbrace{\idotsint}_{y_{j_{min},j=1,\ldots,n}}^{y_{j_{max}}}\underbrace{dy_{1}~dy_{2}~\ldots~dy_{n}}.
  												\label{eq.10}
  \end{equation}
  Now, the volume of $n-$dimensional sphere $(n\geq~3)$ is defined as
  \begin{equation}
  \centering
  S(0,\sqrt{\delta_{i}})=:\underbrace{\idotsint}_{y_{j_{min},j=1,\ldots,n}}^{y_{j_{max}}}\underbrace{dy_{1}~dy_{2}~\ldots~dy_{n}}.
  												\label{eq.11}
  \end{equation}
  with center at origin and radius $\sqrt{\delta_{i}}$. \\
   Using (\ref{eq.10}) and (\ref{eq.11}), we get (\ref{eq.6}). \\
  After obtaining region of attraction (volume) corresponding to each of the driver nodes, (ii) implies directly. \\
  This completes the proof of lemma $1$.
  \end{proof}
  \begin{corol}
  	For $n=2$, we need to compute area. The area of ellipse (region) is given by
  	\begin{equation}
  	\centering
  	Area_{i}=\frac{\pi~\delta_{i}}{\sqrt{det(P_{i})}}
  													\label{eq.12}
  	\end{equation}
  	$Area_{i}>~Area_{j}\Longleftrightarrow~\delta_{j}^{2}det(P_{i})~<~\delta_{i}^{2}det(P_{j})$, $i \neq j =1,2$. The $i^{th}$ node is chosen as the driver node of the system (\ref{eq.1}), corresponding to the largest area.
  	\end{corol}
   We summarize computation of $\delta_{i}$ and $Area_{max}$~/~$Vol_{max}$ for a general $n$ nodes network with $p$ nodes being (each of the nodes can control the network) driver nodes.
   \begin{algorithm}[H]
   \caption{Find maximum area/volume of the ellipsoid and driver node of a network with all eigenvalues in RHP}\label{euclid}
   \textbf{Input} : Network $G:=(V,E)$, Number of nodes $n$, driver nodes $p$ (These driver nodes can control the network independently), $Q>0$ and $R>0$.\\
   \textbf{Output} :  Maximum Area/Volume of the ellipsoid and  corresponding $i^{th}$ driver node
   \begin{algorithmic}[1]
   \FOR{$i=1:p$}
   \STATE Select $B_{i} = [0,0, \ldots , b_{i}, \ldots, 0,0]^{T}$ \\
   \STATE Obtain the positive definite matrix $P_{i}>0$ using (\ref{eq.3})
   \begin{equation}
   \centering
   A^{T}P_{i}+P_{i}A-P_{i}B_{i}R^{-1}B_{i}^{T}P_{i}+Q=0   \nonumber
   \end{equation}
   where $i=1,2,\ldots,p$, $j=1,2,\ldots,n$. \\
   \STATE Calculate the radius of the ellipsoid $\delta_{i}$ as \\
    $\delta_{i}= \frac{4}{\Sigma_{i}}$ \\
    where $\Sigma_{i} = B_{i}^{T}P_{i}B_{i} = b_{i}^{2}p_{jj}^{(i)}$ \\
   \ENDFOR
   \IF {$n=2$}
   \STATE Calculate $Area_{max}$ as \\
    $Area_{max}=arg~ max_{i} \{Area_{i}~|~ i=1,2,\ldots,p\}$ \\
    where $Area_{i}=\frac{\pi \delta_{i}}{\sqrt{det(P_{i})}}$ \\
   \ELSE
   \STATE Calculate $Vol_{max}$ as \\
   $Vol_{max}=arg~ max_{i} \{Vol_{i}~|~ i=1,2,\ldots,p\}$ \\
   where $Vol_{i}=\frac{S(0,\sqrt{\delta_{i}})}{\sqrt{det(P_{i})}}$ \\
   where $S(0,\sqrt{\delta_{i}})$ is the volume of $n-$dimensional sphere $(n\geq~3)$ of radius $\sqrt{\delta_{i}}$ and center at origin, $det(P_{i})$ is the determinant of the positive definite matrix $P_{i}$.
   \ENDIF
   \RETURN{} $Area_{max}$~/~$Vol_{max}$ and node $i$ for which $Area_{max}$ ~/~ $Vol_{max}$ 
   \end{algorithmic}
   \end{algorithm}
   
 \subsection{Region of Attraction and Selection of Driver Node of a Network with Stable and Unstable Eigenvalues}
 Now, we consider a more general network which has stable and unstable eigenvalues. We transform the original system (\ref{eq.1}) into two subsystems. The first subsystem has all unstable eigenvalues and the second subsystem has all stable eigenvalues. In what follows, we consider subsystem with unstable eigenvalues and obtain region of attraction. In this case, the system will be diagonalized using a transformation matrix $V$. Let $\textbf{x}=V\textbf{z}$, the transformation matrix $V$ is the matrix which has eigenvectors of the matrix $A(G)$. \\
   The transformed system has an anti-stable subsystem as well as a stable subsystem of the original system (\ref{eq.1}) defined as
   \begin{eqnarray}
   \centering
   \frac{d\textbf{z}(t)}{dt}=\tilde{A}\textbf{z}(t)+\tilde{B}_{i}sat(u_{i}(t)) 													\label{eq.13}
   \end{eqnarray}
   Where
   \begin{equation}
   \begin{aligned}                              			   																	\label{eq.14}
   \tilde{A} =V^{-1}AV = \begin{bmatrix}
   \tilde{A}_{1} &  0  \\[0.3em]
   0 & \tilde{A}_{2}  \\[0.3em]
            \end{bmatrix}
   \end{aligned}
   \end{equation}
   where $ \tilde{A}_{1} \in \mathbb{R}^{k \times k}$ anti-stable matrix and  $ \tilde{A}_{2} \in \mathbb{R}^{n-k \times n-k}$ stable matrix, and the control input matrix $ \tilde{B}_{i} $ can be partitioned as
   \begin{equation}
   \begin{aligned}                						\label{eq.15}
   \tilde{B}_{i} = V^{-1}B_{i} = \begin{bmatrix}
   \tilde{B}_{1i}   \\[0.3em]
   \tilde{B}_{2i}  \\[0.3em]
   \end{bmatrix}
   \end{aligned}
   \end{equation}
    where $\tilde{B}_{1i} \in \mathbb{R}^{k \times 1}$ and $\tilde{B}_{2i} \in \mathbb{R}^{n-k \times 1}$, $n$ and $k$ are the number of nodes (order) in system (\ref{eq.1}) and order of the anti-stable subsystem respectively. \\
   The transformed system (\ref{eq.13}) can be rewritten as
   \begin{equation}
   \begin{aligned}                 						\label{eq.16}
   \begin{bmatrix}
   \frac{d\textbf{z}_{1}(t)}{dt} \\[0.3em]
   \frac{d\textbf{z}_{2}(t)}{dt} \\[0.3em]
   \end{bmatrix}=
   \begin{bmatrix}
   \tilde{A}_{1} &  0  \\[0.3em]
   0 & \tilde{A}_{2}  \\[0.3em]
            \end{bmatrix}
            \begin{bmatrix}
            \textbf{z}_{1}(t) \\[0.3em]
            \textbf{z}_{2}(t) \\[0.3em]
            \end{bmatrix}+
   \begin{bmatrix}
   \tilde{B}_{1i}   \\[0.3em]
   \tilde{B}_{2i}  \\[0.3em]
   \end{bmatrix} sat(u_{i}(t))
   \end{aligned}
   \end{equation}
   Now we describe the anti-stable subsystem of (\ref{eq.16}) as follows
        \begin{eqnarray}
        \centering
        \frac{d\textbf{z}_{1}(t)}{dt}=\tilde{A}_{1}\textbf{z}_{1}(t)+\tilde{B}_{1i} sat(u_{i}(t))               					\label{eq.17}
        \end{eqnarray}
   \begin{rem}
   To obtain region of attraction, only anti-stable subsystem of (\ref{eq.17}) needs to be considered. For open-loop stable network, global stabilization is possible with saturated driver nodes.
   \end{rem}
  The control law $u_{i}(t)$ can be defined as
   \begin{equation}
   \centering
   u_{i}(t)=-H\textbf{z}_{1}                 \label{eq.18}
   \end{equation}
   where $H=R_{1}^{-1}\tilde{B}_{1i}^{T}\tilde{P}_{1i}$ is a feedback gain matrix. \\
   The positive definite matrix $\tilde{P}_{1i}>0$ is the solution to the following Algebraic Ricatti Equation
    \begin{equation}
    \centering
    \tilde{A}_{1}^{T}\tilde{P}_{1i}+\tilde{P}_{1i}\tilde{A}_{1}-\tilde{P}_{1i}\tilde{B}_{1i}R_{1}^{-1}\tilde{B}_{1i}^{T}\tilde{P}_{1i}+Q_{1}=0                             						\label{eq.19}
    \end{equation}
    where $Q_{1}>0$ and $R_{1}>0$ are chosen positive definite and diagonal matrices. \\
    The radius ($\tilde{\delta}_{1i}$) of the invariant ellipsoid ($\varepsilon_{1}(\tilde{P}_{1 i},\tilde{\delta}_{1i})=\{\textbf{z}_{1}:\textbf{z}_{1}^{T}\tilde{P}_{1 i}\textbf{z}_{1} \leq \tilde{\delta}_{1i}\}$) of the anti-stable network (\ref{eq.17}) is defined as
    \begin{equation}
     	\begin{aligned}
     	\tilde{\delta}_{1i} = \frac{4}{\Sigma_{1i}}      \label{eq.20}
     	\end{aligned}
     	\end{equation}
     	where $\Sigma_{1i}:= \tilde{B}_{1i}^{T}\tilde{P}_{1i}\tilde{B}_{1i}$, ~ $i=1,2,\ldots,p$.
    \begin{lm} For the anti-stable subsystem defined in (\ref{eq.17}) with control law (\ref{eq.18}) and radius of invariant ellipsoid (\ref{eq.20}) \\
           (i) The volume of the invariant ellipsoid is given by \\
           	\begin{equation}
           	\begin{aligned}    						\label{eq.21}
           Vol_{1i}=\frac{S(0,~\sqrt{\tilde{\delta}_{1i}})}{\sqrt{det(\tilde{P}_{1i})}}
           	\end{aligned}
           	\end{equation}
                Where $S(0,\sqrt{\tilde{\delta}_{1i}})$ is the volume of $k-$dimensional sphere $(k\geq~3)$ of radius $\sqrt{\tilde{\delta}_{1i}}$ and center at origin, $det(\tilde{P}_{1i})$ is the determinant of the positive definite matrix $\tilde{P}_{1i}$. \\
            (ii) Largest region of attraction corresponds to $i^{th}$ input for which volume becomes largest i.e. \\
                	$Largest~~ volume=Max_{i}\{Vol_{1i}\} ~~~ i=1,2,\ldots,p$
            \end{lm}
            \begin{proof}
            The proof of lemma $2$ is similar to given in lemma $1$.
            \end{proof}
      
        We summarize computation of $\tilde{\delta}_{1i}$ and $Area_{1,max}$~/~ $Vol_{1,max}$ for a general network of order $n$ with dimension of anti-stable subsystem $k$, $k<n$ in the algorithm $2$.    
      \begin{algorithm}[H]
      \caption{Find maximum area/volume of the invariant ellipsoid and driver nodes of a network system (\ref{eq.1})}  \label{euclid1}
      \textbf{Input} : Network $G:=(V,E)$, Number of nodes $n$, Number of anti-stable nodes $k$, driver nodes $p$, $Q_{1}>0$ and $R_{1}>0$. \\
      \textbf{Output} : Maximum area/volume of the ellipsoid and corresponding $i^{th}$ driver node
      \begin{algorithmic}[1]
      \STATE Partition the matrix $\tilde{A}$ using (\ref{eq.14})
      \begin{equation}
      \begin{aligned}  \nonumber
      \tilde{A} =V^{-1}AV = \begin{bmatrix}
      \tilde{A}_{1} &  0  \\[0.3em]
      0 & \tilde{A}_{2}  \\[0.3em]
               \end{bmatrix}
      \end{aligned}
      \end{equation}
      where $\tilde{A}_{1} \in \mathbb{R}^{k \times k}$ and $\tilde{A}_{2} \in \mathbb{R}^{n-k \times n-k}$.
      \FOR { $i=1:p$ }
      \STATE Select $B_{i} = [0,0, \ldots , b_{i}, \ldots, 0,0]^{T}$ of the original system (\ref{eq.1}) \\
      \STATE Partition $\tilde{B}_{i}$ using (\ref{eq.15}), given as
      \begin{equation}
      \begin{aligned} \nonumber
      \tilde{B}_{i} = V^{-1}B_{i} = \begin{bmatrix}
      \tilde{B}_{1i}   \\[0.3em]
      \tilde{B}_{2i}  \\[0.3em]
      \end{bmatrix}
      \end{aligned}
      \end{equation}\\
      where $\tilde{B}_{1i} \in \mathbb{R}^{k \times 1}$ and $\tilde{B}_{2i} \in \mathbb{R}^{n-k \times 1}$.
      \STATE Calculate matrix $\tilde{P}_{1 i}$ using (\ref{eq.19}), given as
      \begin{equation}
      \centering
      \tilde{A}_{1}^{T}\tilde{P}_{1 i}+\tilde{P}_{1 i}\tilde{A}_{1}-\tilde{P}_{1 i}\tilde{B}_{1 i}R_{1}^{-1}\tilde{B}_{1 i}^{T}\tilde{P}_{1 i}+\tilde{Q}_{1}=0   \nonumber
      \end{equation}
      \STATE Calculate the radius of the invariant ellipsoid $\tilde{\delta}_{1i}$ as \\
          $\tilde{\delta}_{1i}= \frac{4}{\Sigma_{1i}}$ \\
          where $\Sigma_{1 i}=\tilde{B}_{1 i}^{T}\tilde{P}_{1 i}\tilde{B}_{1 i}$ \\
      \ENDFOR
      \IF {$k=2$}
       \STATE Calculate $Area_{1,max}$ as \\
        $Area_{1,max}=arg~ max_{i} \{Area_{1i}~|~ i=1,2,\ldots,p\}$ \\
        where $Area_{1i}=\frac{\pi \tilde{\delta}_{i}}{\sqrt{det(\tilde{P}_{i})}}$ \\
       \ELSE
      \STATE Calculate $Vol_{1,max}$ as \\
       $Vol_{1,max}=arg~ max_{i} \{Vol_{1i}~|~ i=1,2,\ldots,p\}$ \\
       where $Vol_{1i}=\frac{S(0,~\sqrt{\tilde{\delta}_{1i}})}{\sqrt{det(\tilde{P}_{1i})}}$ \\
       where $S(0,\sqrt{\tilde{\delta}_{1i}})$ is the volume of $k-$dimensional sphere $(k\geq~3)$ of radius $\sqrt{\tilde{\delta}_{1i}}$ and center at origin, $det(\tilde{P}_{1i})$ is the determinant of the positive definite matrix $\tilde{P}_{1i}$.
       \ENDIF
       \RETURN{} $Area_{1,max}$~/~$Vol_{1,max}$ and node $i$ for which $Area_{1,max}$~/~$Vol_{1,max}$
      \end{algorithmic}
      \end{algorithm}
 \begin{corol}
             	For $k=2$, the region becomes area. The area of invariant ellipse (region) of system defined in (\ref{eq.17}) is given by
             		\begin{equation}
             		\centering
             		Area_{1i}=\frac{\pi~\tilde{\delta}_{1i}}{\sqrt{det(\tilde{P}_{1i})}}
             									\label{eq.22}
             		\end{equation}
             		$Area_{1i}>~Area_{1j}\Longleftrightarrow~\tilde{\delta}_{1j}^{2}det(\tilde{P}_{1i})~<~\tilde{\delta}_{1i}^{2}det(\tilde{P}_{1j})$, $i \neq j =1,2$, then $i^{th}$ node will be chosen as the driver node of the system (\ref{eq.1}) corresponding to largest area.
             	\end{corol}
      	
	\section{Simulation Results and Discussion}
	In this section, We consider two numerical examples to verify the theoretical results. First example considers anti-stable network which has all positive eigenvalues and other example has stable and unstable eigenvalues.
	  	\begin{example}
	 	We consider the weighted directed network of two nodes system as shown in Fig. $1$.
	 			\begin{figure}[!hbtp]
	 			\centering
	 			\includegraphics[width=6cm,height=2cm]{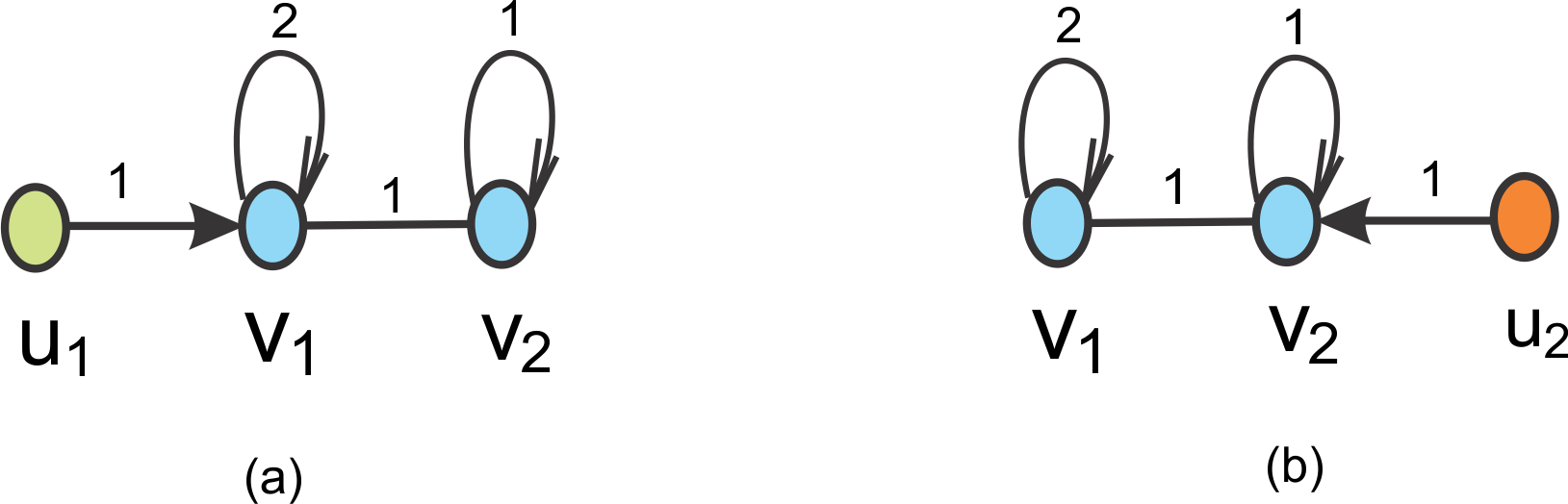}
	 			\captionsetup{format=plain,justification=raggedright}
	 			\caption{Two nodes network system (a) with input  $u_{1}$, (b) with input $u_{2}$.}
	 			\label{fig:1}
	 			\end{figure}
	 	\end{example}	
	     The adjacency and control input matrices are given as:
	 	\begin{equation}
	 	A = \begin{bmatrix}
	 	2 & 1  \\[0.3em]
	 	1 & 1  \\[0.3em]
	 	         \end{bmatrix}, ~~
	 	         B_{1} = \begin{bmatrix}
	 	         	1  \\[0.3em]
	 	         	0  \\[0.3em]
	 	         	         \end{bmatrix}, ~~
	 	         B_{2} = \begin{bmatrix}
	 	         	0  \\[0.3em]
	 	         	1  \\[0.3em]
	 	         	     \end{bmatrix}
	 	         	     						\label{eq.23}
	 	   \end{equation}
	 	\par Let us assume that $B_{i}$ is a variable input matrix which depends on $i^{th}$ driver node in order to control the system. The eigenvalues of the system that can be denoted as $\lambda(A)=\{0.3820,2.6180\}$. Let $Q=I_{2},~ R=1$ and solve the Algebraic Riccati Equation (\ref{eq.3}) for different input matrices $B_{i}$. We obtain the radius of the invariant ellipsoids $\delta_{1}=0.6264$ and $\delta_{2}=0.6068$ using (\ref{eq.5}), for driver node $1$ and $2$ respectively. We get the values of $\sqrt{det(P1)}=7.9030$, $\sqrt{det(P2)}=9.4257$, and also obtain the area of the two ellipsoids using (\ref{eq.12}), they are $0.2473$ Sq. units and $0.2014$ Sq. units respectively. Now we obtain the maximum area of the invariant ellipsoid using lemma $1$ is $0.2473$ Sq. units.
	  The region of attractions of the network considered in example $1$ for different driver node, are plotted in the Fig. $3$. From Fig. $3$, we get the area of first ellipsoid (blue color) which is larger than a second ellipsoid (black color). We conclude here, choosing the first node as the driver node is advantageous because we get the maximum area of the invariant ellipsoid. According to table $1$, results clearly validate the effectiveness of the theoritical results.

\begin{table}[!hbtp]
\centering
\caption{Area of the Contractively Invariant Ellipsoids of different driver nodes}
\begin{tabular}{|c|c|c|}
\hline Driver node & Area of the Contractively & $\frac{Area_{i}}{Area_{max}}$  \\
~& Invariant Ellipsoid & ~\\
\hline 1 & 0.2490 & 1 \\ 
\hline 2 & 0.2023 & 0.8124 \\ 
\hline
\end{tabular}
\end{table} 

\begin{example}
 		We consider the weighted directed network with four nodes system as shown in Fig. $2$.
 					\begin{figure}[!hbtp]
 					\centering
 					\includegraphics[width=5cm,height=4cm]{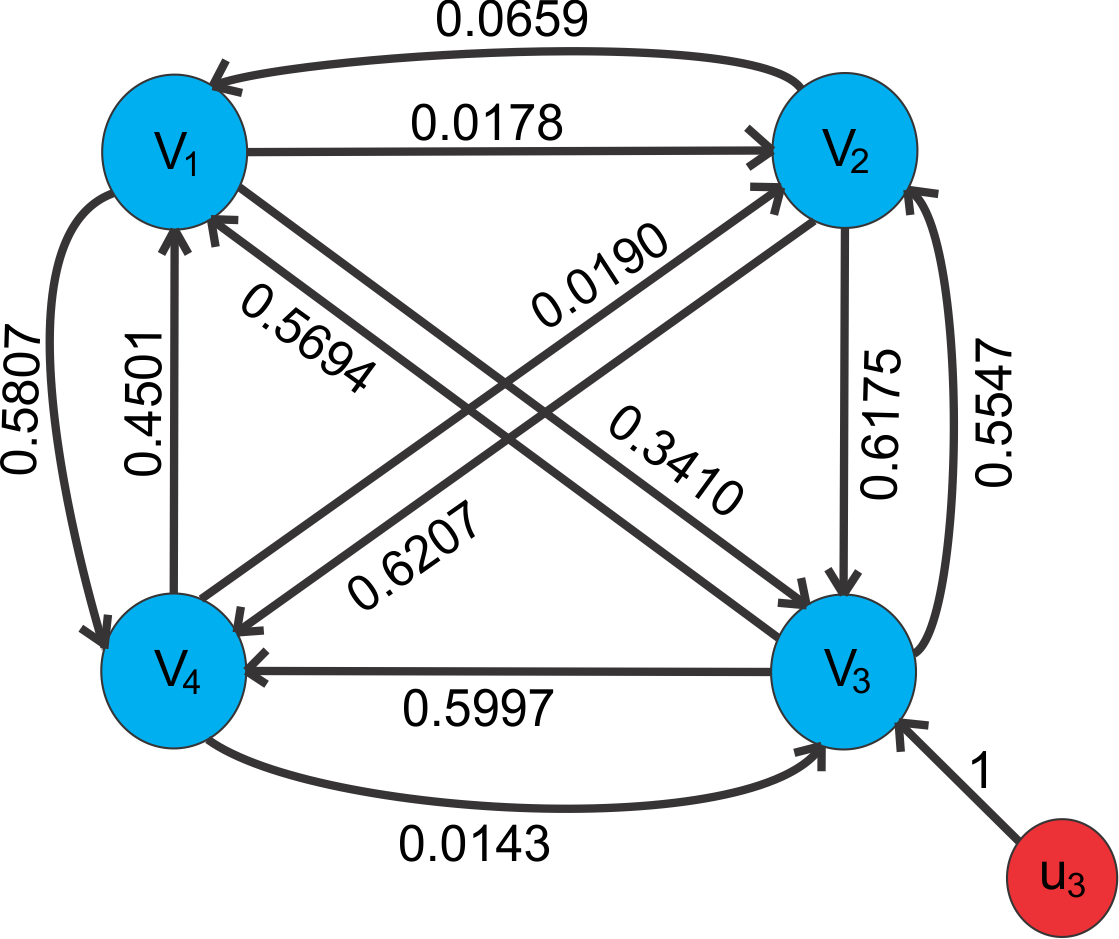}
 					\captionsetup{format=plain,justification=raggedright}
 					\caption{Four nodes network system with control input $u_{3}$. }
 					\label{fig:2}
 					\end{figure}
 		\end{example}			
 		The adjacency matrix and control input matrices are given as:
 		\begin{equation}
 		 		 		\begin{aligned}
 		 		 		A = \begin{bmatrix}
 		 		 		0 & 0.0178 & 0.3410 & 0.5807  \\[0.3em]
 		 		 		0.0659 & 0 & 0.6175 & 0.6207  \\[0.3em]
 		 		 		0.5694 & 0.5547 & 0 & 0.5997  \\[0.3em]
 		 		 	    0.4501 & 0.0190 & 0.0143 & 0  \\[0.3em]
 		 		 		         \end{bmatrix}  \label{eq.24}
 		 		 		\end{aligned}
 		 		 		\end{equation}
 		 		 		\begin{equation}
 		 		 				\begin{aligned}
 		 		 	 B_{1} = \begin{bmatrix}
 		 		 	  	1  \\[0.3em]
 		 		 		0  \\[0.3em]
 		 		 		0  \\[0.3em]
 		 		 		0  \\[0.3em]
 		 		          \end{bmatrix},~~~
 		 		 		         B_{2} = \begin{bmatrix}
 		 		 		         	0  \\[0.3em]
 		 		 		         	1  \\[0.3em]
 		 		 		         	0  \\[0.3em]
 		 		 		         	0  \\[0.3em]
 		 		 		         	     \end{bmatrix}, ~~~
 		 		 		         B_{3} = \begin{bmatrix}
 		 		 		            0  \\[0.3em]
 		 		 		         	0  \\[0.3em]
 		 		 		         	1  \\[0.3em]
 		 		 		         	0  \\[0.3em]
 		 		 		                \end{bmatrix}, ~~~
 		 		 		         B_{4} = \begin{bmatrix}
 		 		 		            0  \\[0.3em]
 		 		 		         	0  \\[0.3em]
 		 		 		         	0  \\[0.3em]
 		 		 		         	1  \\[0.3em]
 		 		 		           	     \end{bmatrix}   \label{eq.25}
 		 		 		      \end{aligned}
 		 		 		     \end{equation}
 Depending on choice of a particular node as a driver node, input matrix $B_{i}$ has four possibilities. The eigenvalues of the network shown in Fig.$2$, obtained as $\lambda(A)=\{0.9613,0.1318,-0.7706,-0.3225\}$. Here two eigenvalues are positive so the network will be partitioned into two subsystem according to (\ref{eq.16}). Now we analyse anti-stable subsystem (\ref{eq.17}). Let $Q=I_{4},~ R_{1}=1$ and solve the Algebraic Riccati Equation (\ref{eq.19}) for different input matrices $B_{i}$ corresponding to different driver nodes. We obtain the radius of the invariant ellipsoids $\tilde{\delta}_{11}=1.2289$, $\tilde{\delta}_{12}=1.4606$, $\tilde{\delta}_{13}=1.6647$ and $\tilde{\delta}_{14}=1.3311$  using $(\ref{eq.20})$. We get the values of $\sqrt{det(\tilde{P}_{11})}=4.8641$, $\sqrt{det(\tilde{P}_{12})}=12.0531$, $\sqrt{det(\tilde{P}_{13})}=13.3164$ and $\sqrt{det(\tilde{P}_{14})}=4.4125$ with $B_{i}$ respectively.
 \par Now from the table $1$, we obtain the maximum area of the invariant ellipsoid $0.9449$ Sq. units, corresponding to driver node $4$. We observe that, the importance of driver node from network stability point of view depends upon the area of the ellipsoids. The region of attractions of the network considered in example $2$ for different driver nodes are plotted in the Fig. $4$.
 		
 From Fig. $4$, we get the largest area when fourth node is selected as input (magenta color) vis-a-vis region corresponding to node $1$, $2$, $3$ (blue,black and red colors). We conclude here, choosing the fourth node as the driver node is advantageous because we get the maximum area of the invariant ellipsoid. According to table $2$, results clearly validate the effectiveness of the theoretical results.
 		\begin{table}
 		\centering
 		\captionsetup{format=plain,justification=raggedright}
 		\caption{Area of the Contractively Invariant Ellipsoids of different driver nodes}
 		\small\addtolength{\tabcolsep}{5pt}
 		\begin{tabular}{|c|c|c|}
 		\hline Driver node & Area of the Contractively & $\frac{Area_{i}}{Area_{max}}$  \\
 		~& Invariant Ellipsoid & ~\\
 		\hline 4 & 0.9449 & 1 \\
 		\hline 1 & 0.7917 & 0.8378\\
 		\hline 3 & 0.3910 & 0.4138 \\
 		\hline 2 & 0.3796 & 0.4017 \\
 		\hline
 		\end{tabular}
 		\end{table}

\begin{figure}[!hbtp]
 	\centering
 	\includegraphics[width=7cm,height=3cm]{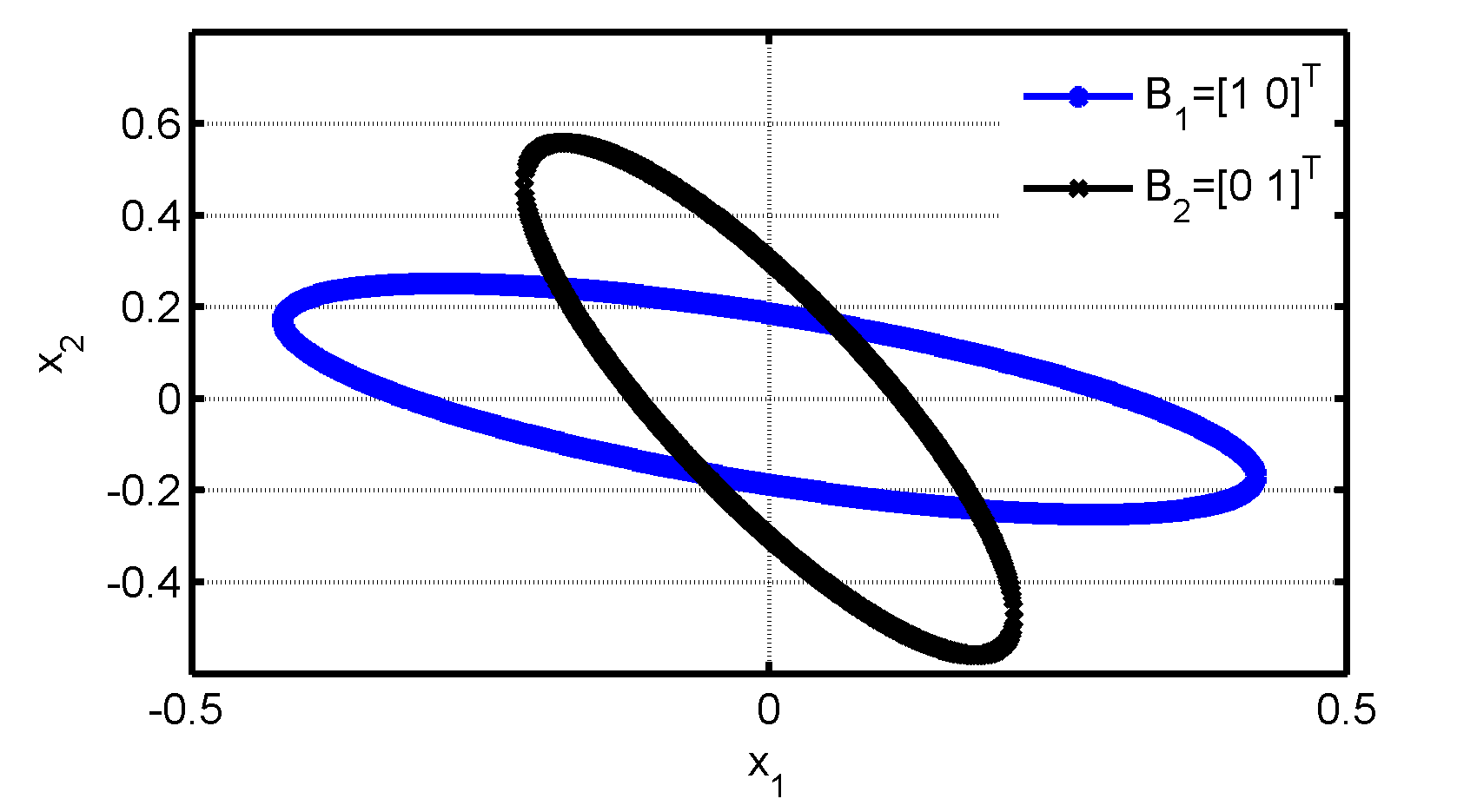}
 	\captionsetup{format=plain,justification=raggedright}
 	\caption{Region of attractions of the network considered in example $1$ for different control input matrices $B_{i}$.}
 			\label{fig:3}
 	\end{figure}

 	\begin{figure}[!hbtp]
 	\centering
 	\includegraphics[width=8.5cm,height=5cm]{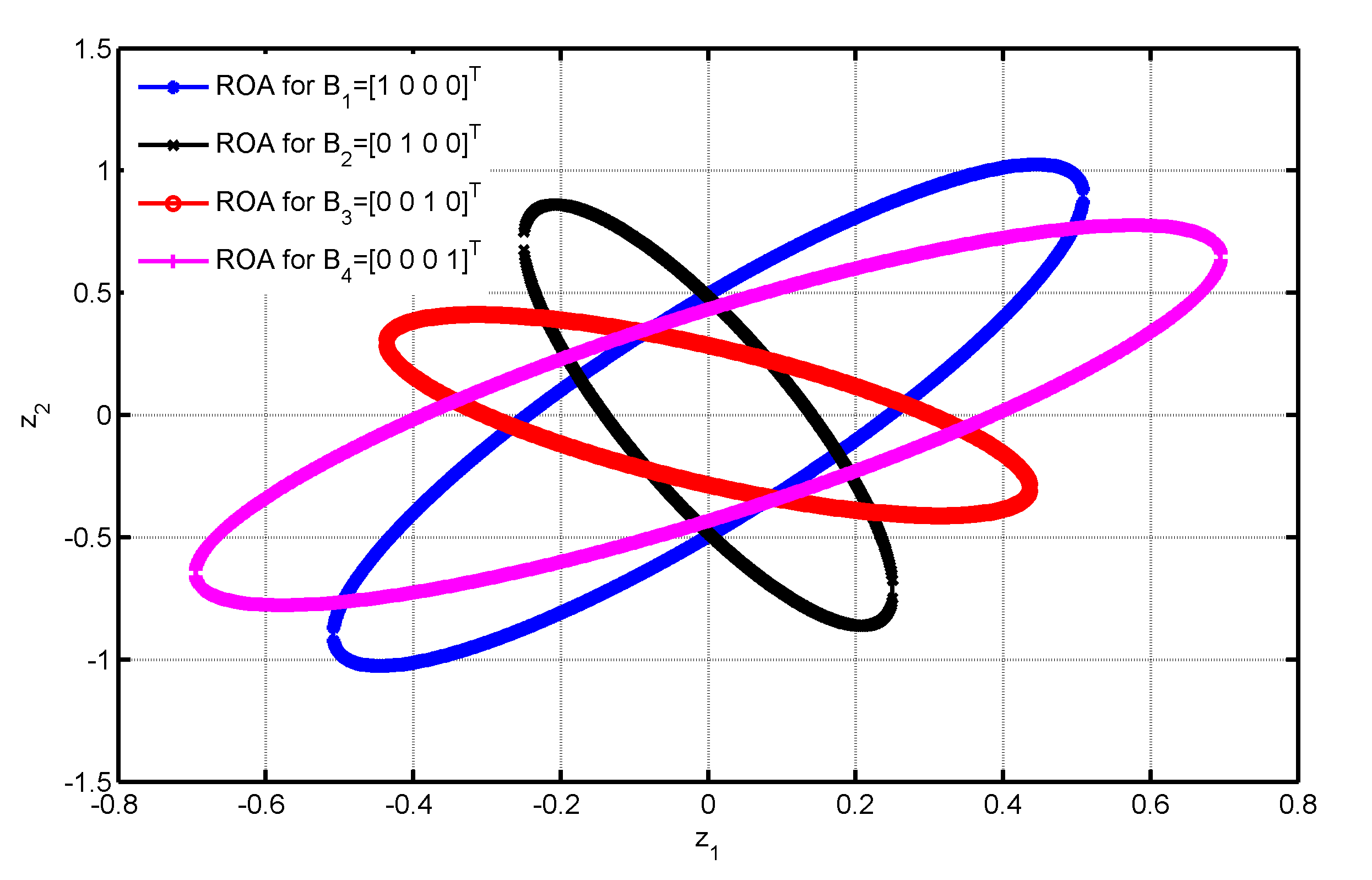}
 	\captionsetup{format=plain,justification=raggedright}
 	\caption{Region of attractions of the network considered in example $2$ for different control input matrices $B_{i}$.}
 								\label{fig:5}
 	\end{figure}

\section{Conclusions}
 The paper has presented a method to select driver nodes. A practical limitation on actuator magnitude has been considered. The proposed theory and algorithms give most suitable driver nodes so that stability with saturated driver node can be ensured in the largest possible region. It has been shown, using numerical example, $40\%$ reduction in stability region occurs if driver node is not selected using the proposed algorithm. A close agreement between theoretical and simulation has been obtained. The determinant of positive definite matrix $P_{i}$ and the radius of the invariant ellipsoids decides which node should be considered as a driver node to control the entire system.






%

\bibliographystyle{IEEEtran}
\bibliography{IEEEabrv,ref}

%
%

\end{document}